\documentclass[aps,pra,twocolumn,nofootinbib, superscriptaddress,10pt]{revtex4-1}

\usepackage{mathdots}
\usepackage{amsfonts,amssymb,amsmath}
\usepackage{mathrsfs}
\usepackage[]{graphics,graphicx,epsfig}
\usepackage{amsthm}
\usepackage{csquotes}
\MakeOuterQuote{"}
\usepackage{epstopdf}
\usepackage{color}
\usepackage[bookmarks=false]{hyperref}

\def\identity{\leavevmode\hbox{\small1\kern-3.8pt\normalsize1}}

\newtheorem{lemma}{Lemma}
\newtheorem{proposition}{Proposition}

\newcommand{\ket}[1]{\left | #1 \right\rangle}

\newcommand{\diag}{\operatorname{diag}}

\newcommand{\tr}{\operatorname{tr}}
\newcommand{\1}{\operatorname{\uppercase\expandafter{\romannumeral1}}}
\newcommand{\2}{\operatorname{\uppercase\expandafter{\romannumeral2}}}
\newcommand{\3}{\operatorname{\uppercase\expandafter{\romannumeral3}}}
\newcommand{\4}{\operatorname{\uppercase\expandafter{\romannumeral4}}}
\newcommand{\5}{\operatorname{\uppercase\expandafter{\romannumeral5}}}
\newcommand{\6}{\operatorname{\uppercase\expandafter{\romannumeral6}}}
\newcommand{\7}{\operatorname{\uppercase\expandafter{\romannumeral7}}}
\newcommand{\8}{\operatorname{\uppercase\expandafter{\romannumeral8}}}
\newcommand{\9}{\operatorname{\uppercase\expandafter{\romannumeral9}}}

\newcommand{\GHZ}{\mathrm{G}}

\newcommand{\rmi}{\mathrm{i}}

\newcommand{\rme}{\operatorname{e}}

\newcommand{\caH}{\mathcal{H}}
\newcommand{\caP}{\mathcal{P}}
\newcommand{\caB}{\mathcal{B}}
\newcommand{\bbZ}{\mathbb{Z}}

\newcommand{\bfr}{\mathbf{r}}
\newcommand{\bfh}{\mathbf{h}}
\newcommand{\bfj}{\mathbf{j}}
\newcommand{\bfg}{\mathbf{g}}
\newcommand{\scrX}{\mathscr{X}}
\newcommand{\scrY}{\mathscr{Y}}
\newcommand{\scrZ}{\mathscr{Z}}

\renewcommand{\epsilon}{\varepsilon}

\def\eqref#1{\textup{(\ref{#1})}}
\newcommand{\tref}[1]{Table~\ref{#1}}
\newcommand{\eref}[1]{Eq.~\textup{(\ref{#1})}}
\newcommand{\lref}[1]{Lemma~\ref{#1}}
\newcommand{\pref}[1]{Proposition~\ref{#1}}
\newcommand{\eqsref}[2]{Eqs.~(\ref{#1}) and (\ref{#2})}

\def\<{\langle}  
\def\>{\rangle}  

\newcommand{\rcite}[1]{Ref.~\cite{#1}}
\newcommand{\rscite}[1]{Refs.~\cite{#1}}

\begin{document}
\title{Optimal Verification of Greenberger-Horne-Zeilinger States}

\author{Zihao Li}
\affiliation{Department of Physics and Center for Field Theory and Particle Physics, Fudan University, Shanghai 200433, China}
\affiliation{State Key Laboratory of Surface Physics, Fudan University, Shanghai 200433, China}

\author{Yun-Guang Han}
\affiliation{Department of Physics and Center for Field Theory and Particle Physics, Fudan University, Shanghai 200433, China}
\affiliation{State Key Laboratory of Surface Physics, Fudan University, Shanghai 200433, China}

\author{Huangjun Zhu}
\email{zhuhuangjun@fudan.edu.cn}
\affiliation{Department of Physics and Center for Field Theory and Particle Physics, Fudan University, Shanghai 200433, China}
\affiliation{State Key Laboratory of Surface Physics, Fudan University, Shanghai 200433, China}
\affiliation{Institute for Nanoelectronic Devices and Quantum Computing, Fudan University, Shanghai 200433, China}
\affiliation{Collaborative Innovation Center of Advanced Microstructures, Nanjing 210093, China}

\begin{abstract}
We construct optimal protocols for verifying  qubit and qudit GHZ states using local projective measurements. When the local dimension is a prime, an optimal protocol is constructed from Pauli measurements only.
Our protocols provide a highly efficient way for estimating the fidelity and certifying genuine multipartite entanglement.  In particular, they enable the certification of genuine multipartite entanglement  using only one test when the local dimension is sufficiently large.
By virtue of adaptive local projective measurements, we then construct protocols for verifying GHZ-like states that are optimal over all protocols based on one-way communication.  The efficiency can be improved further if additional communications are allowed.  Finally, we construct optimal protocols for verifying GHZ states and nearly optimal protocols for  GHZ-like states in the adversarial scenario.
\end{abstract}

\date{\today}
\maketitle

\section{Introduction}
Greenberger-Horne-Zeilinger (GHZ) states \cite{GHZ89,GHZ90} are typical examples of quantum states with genuine multipartite entanglement (GME) \cite{Guhne09}. They play key roles both  in quantum information processing and in foundational studies, such as quantum secret sharing \cite{Titt01,Chen05}, entanglement purification \cite{Cheong07},
open-destination teleportation \cite{Zhao04}, quantum networks \cite{McCut16}, randomness verification \cite{HayashiK18}, and multipartite nonlocality tests \cite{Scara01,Zhang15}.
The significance of GHZ states is witnessed by numerous experiments devoted to preparing them in various platforms, with ever-increasing number of particles \cite{Zhang15,Wang2016,Song17,Erhard18,Ji19,Resch05,Cruz18,Imany19}. In practice, multipartite quantum states
prepared in experiments are never perfect, so  it is crucial to verify these states with high precision using limited resources.
However, traditional tomographic approaches   are known to be resource consuming and very inefficient \cite{Resch05,Cruz18,Haff05}.
Even with popular alternatives like direct fidelity estimation \cite{FlamL11}, the scaling behaviors of the number of required measurements  with the infidelity and the qubit number are suboptimal.

Recently, an alternative approach known as quantum state verification has attracted increasing attention \cite{HayaMT06,Aolita15,Hang17,PLM18,ZhuEVQPSshort19,ZhuEVQPSlong19}.
Efficient verification protocols based on local  operations and classical communication (LOCC) have been constructed for stabilizer states \cite{HayaM15,PLM18,ZhuH19E,ZhuEVQPSlong19,Kalev19}, hypergraph states  \cite{ZhuH19E}, and Dicke states \cite{Liu19}. However, optimal protocols are known only for
maximally entangled states \cite{HayaMT06, Haya09,ZhuH19O} and bipartite pure states under restricted LOCC \cite{PLM18,LHZ19,Wang19,Yu19}. For quantum states with GME, such as GHZ states,
no optimal protocol has been found  so far because such optimization problems are usually extremely difficult. In addition, most protocols known so far are not homogeneous, which is not desirable for practical applications \cite{ZhuEVQPSlong19}.
Any progress on these issues is of interest to both theoretical studies and practical applications.

In this paper, we propose optimal  protocols for verifying (qubit and qudit) GHZ states using local projective measurements.
When the local dimension is a prime, only Pauli measurements are required. Moreover, all the protocols we construct  are  homogeneous. They offer a highly  efficient tool for  fidelity estimation and entanglement certification. Surprisingly, the GME can be certified with any given significance level using only one test when the local dimension is sufficiently large, which has never been achieved or even anticipated before. By virtue of adaptive local projective measurements, our protocols can be generalized to GHZ-like states, while retaining the high efficiency. Moreover, these protocols can be applied to the adversarial scenario with minor modification. Now the protocols for verifying GHZ states based on local projective measurements are actually optimal among all possible protocols without locality restriction.
Besides quantum state verification, our protocols are also useful for verifying quantum gates, including some Clifford gates and the controlled-swap (CSWAP) gate \cite{ZhuZ20}.

\section{pure-state verification}
\subsection{Basic framework}
Before proposing protocols for verifying GHZ states, let us  briefly  review  the general framework of pure-state verification \cite{PLM18,ZhuEVQPSshort19,ZhuEVQPSlong19}.
Consider a quantum device that  is supposed to produce the target state $|\Psi\>\in\caH$, but actually produces the states $\sigma_1,\sigma_2,\dots,\sigma_N$ in $N$ runs.
Our task is to verify  whether these states are sufficiently close to the target state on average.
To achieve this task, we can perform two-outcome  measurements  $\{E_l,\openone-E_l\}$ from a set of accessible measurements (projective measurements are most appealing in practice, but our discussions  apply to general measurements).
Each measurement represents a test, and the test operator  $E_l$ corresponds to passing the test.
Here we require that the target state $|\Psi\>$ can  always pass the test, that is,  $E_l|\Psi\>=|\Psi\>$.
Suppose the test $\{E_l,\openone-E_l\}$ is performed with probability $p_l$, then the verification operator (also called a strategy)
is given by $\Omega=\sum_{l} p_l E_l$.
If  $\<\Psi|\sigma_j|\Psi\>\leq1-\epsilon$, then the average probability that $\sigma_j$ can pass each test satisfies \cite{PLM18,ZhuEVQPSlong19}
\begin{equation}\label{eq:ProbPass1test}
\tr(\Omega \sigma_j)\leq 1- [1-\beta(\Omega)]\epsilon=1- \nu(\Omega)\epsilon,
\end{equation}
where $\beta(\Omega)$ denotes the second largest eigenvalue of $\Omega$, and $\nu(\Omega):=1-\beta(\Omega)$ is the spectral gap from the maximum eigenvalue. The inequality in \eref{eq:ProbPass1test} is saturated when $\<\Psi|\sigma_j|\Psi\>=1-\epsilon$ and $\sigma_j$ is supported on the subspace associated with
the two largest eigenvalues of $\Omega$.

Suppose the states $\sigma_1,\sigma_2,\dots,\sigma_N$ are  independent  of each other  and let $\epsilon_j=1-\<\Psi|\sigma_j|\Psi\>$.
Then the probability that these states can pass all $N$ tests satisfies the following tight upper bound \cite{ZhuEVQPSshort19,ZhuEVQPSlong19}
\begin{equation}\label{eq:ProbPassNtest}
\prod_{j=1}^N\tr(\Omega \sigma_j) \leq \prod_{j=1}^N[1- \nu(\Omega)\epsilon_j]\leq[1- \nu(\Omega)\bar{\epsilon}]^N,
\end{equation}
where $\bar{\epsilon}=\sum_j\epsilon_j/N$ is the average infidelity.
In order to ensure the condition $\bar{\epsilon}<\epsilon$ with significance level $\delta$, that is, to ensure the condition
$\prod_{j}\tr(\Omega \sigma_j)\leq\delta$ when $\bar{\epsilon}\geq\epsilon$,
it suffices to perform \cite{ZhuEVQPSshort19,ZhuEVQPSlong19}
\begin{equation}\label{eq:NumberTest}
N=\biggl\lceil\frac{ \ln \delta}{\ln[1-\nu(\Omega)\epsilon]}\biggr\rceil\approx \frac{ \ln \delta^{-1}}{\nu(\Omega)\epsilon}
\end{equation}
tests.
To minimize the number of tests, we need to maximize the value of the spectral gap $\nu(\Omega)$ under LOCC. This task is usually extremely difficult if not impossible. It should be pointed out that the approximation in \eref{eq:NumberTest} is valid only when $\nu(\Omega)\epsilon\ll 1$. Otherwise, the minimum number of tests required is more sensitive to the spectral gap $\nu(\Omega)$, which has an important implication for entanglement certification, as we shall  see in Sec.~\ref{sec:GME}.

In the above presentation we follow the assumptions and interpretation in Refs.~\cite{ZhuEVQPSshort19,ZhuEVQPSlong19}, in contrast to the counterpart  in Ref.~\cite{PLM18}. The authors in Ref.~\cite{PLM18} assume that either $\sigma_j=|\Psi\>\<\Psi|$ for all $j$ (good case)
or $\epsilon_j\geq\epsilon$ for all $j$ (bad case), and the task is to distinguish which case occurs.
However, this assumption is difficult to guarantee and is not feasible from a practical point of view.
The assumptions in Refs.~\cite{ZhuEVQPSshort19,ZhuEVQPSlong19} are more reasonable and the conclusion is stronger in comparison. In particular, the average fidelity rather than the maximum fidelity of the prepared states $\sigma_1,\sigma_2,\dots,\sigma_N$  can be verified.

\subsection{\label{sec:Homo}Homogeneous strategies}
A verification strategy is  \emph{homogeneous} if the verification operator $\Omega$ has the following form
\begin{equation}\label{eq:HomoStra}
\Omega=|\Psi\>\<\Psi|+\beta(\Omega)(\openone-|\Psi\>\<\Psi|)
\end{equation}
for some $0\leq\beta(\Omega)<1$.
Homogeneous strategies are most appealing for quantum state verification because of several important merits not shared by inhomogeneous strategies.
To see this, let us consider a simple scenario in which all  $\sigma_1,\sigma_2,\dots,\sigma_N$ are identical to the state $\sigma$ with fidelity $\<\Psi|\sigma|\Psi\>=F=1-\epsilon$.  If $\Omega$ is homogeneous,  then the bounds in
\eqsref{eq:ProbPass1test}{eq:ProbPassNtest} are saturated, so there is a simple connection between the passing probability and the fidelity of the states prepared, namely,
$\tr(\sigma \Omega)=[1-\beta(\Omega)]F+\beta(\Omega)$, which implies that
\begin{equation}\label{eq:FEIhom}
F=\frac{\tr(\Omega\sigma)-\beta(\Omega)}{\nu(\Omega)},\quad 1-F=\frac{1-\tr(\Omega\sigma)}{\nu(\Omega)}.
\end{equation}
Based on this connection, we can estimate the fidelity and infidelity accurately given sufficiently many tests.  According to \rcite{ZhuEVQPSlong19},
the standard deviation of this estimation is
\begin{equation}\label{eq:Fstd}
\Delta F=\frac{\sqrt{(1-F)(F+\nu^{-1}-1)}}{\sqrt{N}}
\leq \frac{1}{2\nu \sqrt{N}},
\end{equation}
where $\nu=\nu(\Omega)=1-\beta(\Omega)$ and $N$ is the number of tests performed.

If the strategy $\Omega$  is inhomogeneous by contrast, given the passing probability
we can only derive lower and upper bounds for the infidelity \cite{ZhuEVQPSlong19}
\begin{equation}\label{eq:FidelityPassingProb2}
\frac{1-\tr(\Omega\sigma)}{1-\tau(\Omega)}\leq 1-F\leq\frac{1-\tr(\Omega\sigma)}{\nu(\Omega)},
\end{equation}
where $\tau(\Omega)$ is the smallest eigenvalue of $\Omega$. The lower bound in \eref{eq:FidelityPassingProb2} is saturated
when  $\sigma$ is supported on the subspace associated with  the largest and the smallest eigenvalues of $\Omega$.
When the verification operator $\Omega$ is singular, that is,  $\tau(\Omega)=0$, the upper bound is $1/\nu(\Omega)$ times as large as the lower bound. When $\nu(\Omega)=0.1$ for example, in the worst case  we can only conclude (with a given significance level) that the infidelity is smaller than 0.1 (0.5)  even if the actual infidelity is only 0.01 (0.05). Such a conclusion is far from being satisfactory even though it is correct.
When the verification protocol is applied to entanglement detection, this problem makes it much more difficult to detect entanglement. Unfortunately,  the problem cannot be resolved by increasing the number of tests. Therefore, it is desirable to construct a homogeneous verification strategy whenever possible.

In addition, homogeneous strategies are appealing for quantum state verification in the adversarial scenario \cite{ZhuEVQPSshort19,ZhuEVQPSlong19}. In particular, a homogeneous strategy is the most efficient among all verification strategies  with the same spectral gap; it can achieve a much better scaling behavior in the number of tests compared with a singular strategy.

\section{Optimal Verification of GHZ states}
Here we are mainly interested in GHZ states of the form \cite{GHZ89,GHZ90}
\begin{equation}\label{eq:GHZqudit}
|\GHZ_n^d\>=\frac{1}{\sqrt{d}}\sum_{j=0}^{d-1}|j\>^{\otimes n}.
\end{equation}
Previously, a coloring protocol was proposed in \rcite{ZhuH19E} (cf.~Ref.~\cite{HayashiK18}), which can achieve a spectral gap of $1/2$ using two settings based on Pauli measurements, but this protocol is not homogeneous (see \tref{tab:ProtocolGHZ}), and the verification operator is singular.
For a bipartite maximally entangled state of the same local dimension,  the maximum  spectral gap of any verification operator based on LOCC (or separable measurements) is $d/(d+1)$ \cite{HayaMT06, PLM18,Haya09,ZhuH19O}. Obviously, the counterpart for GHZ states cannot be larger. Here we shall show that this upper bound can  be saturated.

\newcommand{\cmark}{\ding{51}}%
\newcommand{\xmark}{\ding{55}}%
\begin{table*}
	\caption{\label{tab:ProtocolGHZ}
Comparison of  verification strategies for the $n$-qudit GHZ state $|\GHZ_n^d\>$ in \eref{eq:GHZqudit}.
Here $R_d(\Omega)$ denotes the range of the local dimension over which each strategy is applicable, $\nu(\Omega)$ denotes the spectral gap of each strategy,
$N(\epsilon,\delta,\Omega)$ denotes the number of tests required to verify the target state within infidelity $\epsilon$ and significance level $\delta$, and
$N_{\rm MS}(\Omega)$ denotes the number of potential  measurement settings.
Strategies $\Omega_{\rm PLM}$ and $\Omega_{\rm ZH}$ are proposed in \rscite{PLM18} and \cite{ZhuH19E}, respectively; the other three strategies are proposed in this paper.
}		
\begin{math}
\begin{array}{c|c|c|c|c|c}
\hline\hline
\mbox{Strategy}  &R_d(\Omega) &\nu(\Omega) &\mbox{Is $\Omega$ homogeneous?} &N(\epsilon,\delta,\Omega) & N_{\rm MS}(\Omega) \\[0.5ex]
\hline
\Omega_{\rm PLM}\ \mbox{\cite{PLM18}}   & d=2      &2^{n-1}/(2^n-1) &\mbox{Yes} &(2^n-1)2^{1-n}\epsilon^{-1}\ln\delta^{-1}  & 2^n-1  \\[0.5ex]
\Omega_{\rm ZH} \ \mbox{\cite{ZhuH19E}} & d\geq2           &1/2     &\mbox{No} &      2       \epsilon^{-1}\ln\delta^{-1}  & 2      \\[0.5ex]
\Omega_{\1}                             & d=2              &2/3     &\mbox{Yes} &(3/2)   \epsilon^{-1}\ln\delta^{-1}  & 2^{n-1}+1    \\[0.5ex]		
\Omega_{\2}                             & d \ \,\mbox{is odd prime} &d/(d+1) &\mbox{Yes} &(d+1)d^{-1} \epsilon^{-1}\ln\delta^{-1}  & d^{n-1}+1   \\[0.5ex]		
\Omega_{\3}          & d\geq3           &d/(d+1) &\mbox{Yes} &(d+1)d^{-1}\epsilon^{-1}\ln\delta^{-1}   & \big\lceil\frac{3}{4}(d-1)^2\big\rceil^{n-1}+1 \\[0.5ex]
\hline\hline
\end{array}	
\end{math}
\end{table*}

\subsection{\label{sec:QubitGHZ}Optimal verification of the $n$-qubit GHZ state}

First, we consider optimal verification of the $n$-qubit GHZ state $|\GHZ_n^2\>$ based on Pauli measurements.
 Recall that the Pauli group for each qubit is generated by  three Pauli matrices,
\begin{equation}\label{eq:PauliM}
X=\begin{pmatrix}0&1\\1&0\end{pmatrix},\quad\
Y=\begin{pmatrix}0&-\rmi\\\rmi&0\end{pmatrix},\quad\
Z=\begin{pmatrix}1&0\\0&-1\end{pmatrix}.
\end{equation}
Denote by $I$ the identity operator on the Hilbert space of one party, then
a Pauli measurement is specified by a string  in $\{I, X,Y, Z\}^n$, which determines the Pauli operators measured on individual qubits;
the identity means no measurement. The weight of the Pauli measurement is the number of terms in the string that are not equal to the identity. The Pauli measurement is complete if the weight is equal to $n$, that is, the string does not contain the identity. A test operator $E$ (and the corresponding test) based on a Pauli measurement is \emph{not admissible}  if there exists another test operator $E'$ based on the same or a different Pauli measurement such that $E'\leq E$ and $\tr(E')<\tr(E)$; otherwise, the test operator $E$ (and the corresponding test) is \emph{admissible}. A Pauli measurement is admissible if at least one admissible test operator can be constructed from this Pauli measurement and not admissible otherwise.

Given a Pauli measurement, let $\{\Pi_1,\Pi_2, \ldots, \Pi_q\}$ be the set of projectors corresponding to the measurement outcomes.
The \emph{canonical test projector} is defined as
\begin{equation}
P= \sum_{\<\GHZ_n^2|\Pi_o |\GHZ_n^2\> >0} \Pi_o.
\end{equation}
To guarantee that the target state $|\GHZ_n^2\>$ can always pass the test, any other test operator $E$ satisfies $E\geq P$ and thus cannot be admissible.  The Pauli measurement is admissible iff the canonical test projector is admissible. These observations reveal the crucial role of canonical test projectors in constructing an efficient verification protocol.

For example, the canonical test projector associated with $Z^n$  (understood as a Pauli string with $n$ Pauli operators equal to $Z$) reads
\begin{equation}\label{eq:TestProjZn}
P_0=(|0\>\<0|)^{\otimes n}+(|1\>\<1|)^{\otimes n};
\end{equation}
the test is passed iff the outcomes of all $Z$ measurements on individual qubits coincide. Given a string in $\{X,Y\}^n$, let $\scrY$  be the set of parties that perform $Y$ measurements, then  $\overline{\scrY}:=\{1,2,\ldots, n\}\setminus \scrY$  is the set of parties that perform $X$ measurements. When $|\scrY|=2t$ is even, the canonical test projector reads
\begin{equation}\label{eq:PY}
P_{\scrY}=\frac{1}{2} \biggl[\openone + (-1)^t  \prod_{k\in \scrY} Y_k \prod_{k'\in \overline{\scrY}} X_{k'}\biggr];
\end{equation}
the test is passed iff the total number of outcomes $-1$ (either from $X$ or $Y$ measurements) has the same parity as $t$. The following lemma clarifies all admissible Pauli measurements and test operators for  $|\GHZ_n^2\>$; see Appendix~\ref{app:ProofProposition} for a proof.

\begin{lemma}\label{lem:TestProjAdmGHZ2}
The GHZ state $|\GHZ_n^2\>$ has $1+2^{n-1}$ admissible Pauli measurements, namely $Z^n$ and all strings in $\{X,Y\}^n$ with even numbers of $Y$. The corresponding $1+2^{n-1}$ canonical test projectors in  Eqs.~\eqref{eq:TestProjZn} and \eqref{eq:PY} are the only admissible test operators.
\end{lemma}

Our verification protocol is composed of $1+2^{n-1}$ admissible tests in which the test $P_0$ is performed with probability $1/3$ and the other $2^{n-1}$ tests are performed with  probability $1/(3\times 2^{n-2})$ each. The verification operator reads
\begin{equation}\label{eq:Omega1}
\Omega_{\1} := \frac{1}{3}\biggl(P_0+\frac{1}{2^{n-2}}\sum_{\scrY}P_{\scrY}\biggr) = \frac{1}{3}\big(\openone+ 2|\GHZ_n^2\>\<\GHZ_n^2|\big),
\end{equation}
which is homogeneous.
Here the second equality is proved in Appendix~\ref{app:ProofsEquations}.
We have $\beta(\Omega_{\1})=1/3$, and
\begin{equation}
\nu(\Omega_{\1})=\frac{2}{3},\qquad N(\Omega_{\1})\approx \frac{3}{2\epsilon} \ln \delta^{-1}.
\end{equation}
This protocol is optimal among all protocols based on LOCC or separable measurements.
Compared with the  strategy in \rcite{PLM18} which achieves $\nu=2^{n-1}/(2^n-1)$ with $2^n-1$ measurement settings, our strategy not only has a higher efficiency, but also requires fewer measurement settings, as illustrated in \tref{tab:ProtocolGHZ}. The protocol proposed in Ref.~\cite{ZhuH19E} requires much fewer measurement settings, but it is not homogeneous and thus has a number of drawbacks as mentioned in Sec.~\ref{sec:Homo}.  The current protocol is the most appealing if it is not difficult to switch Pauli measurements, which is the case in most scenarios of practical interest.

Moreover, our protocol proposed above  is essentially the
unique optimal protocol based on Pauli measurements as shown in \pref{pro:GHZ2unique} below and proved in Appendix~\ref{app:ProofProposition}.
In particular, the number $1+2^{n-1}$ of (potential) measurement settings cannot be reduced. In addition, all canonical test projectors are required to construct a homogeneous strategy. It should be pointed out that there is some freedom in choosing the Pauli group:  different choices are related to each other by local unitary transformations. Here we focus on the canonical Pauli group generated by Pauli matrices in \eref{eq:PauliM} for each qubit;
only  nonadaptive Pauli measurements associated with this Pauli group are considered.
Nevertheless, the test operators are not required to be projectors, although it turns out that this relaxation does not provide any advantage.

\begin{proposition}\label{pro:GHZ2unique}
	Suppose  $\Omega$ is an optimal  verification strategy with $\nu(\Omega)=2/3$ for $|\GHZ_n^2\>$ that is based on Pauli measurements. Then $\Omega=\Omega_{\1}$; in addition,
	$\Omega$ is composed of admissible tests with the same probabilities  as in 	 $\Omega_{\1}$.
\end{proposition}

Besides quantum state verification, our protocol is also useful for verifying quantum gates, including Clifford gates and the CSWAP gate.
The basic idea of quantum gate verification is to feed some pure test states into the quantum gate or gate set to be verified, and then verify the output states \cite{ZhuZ20,LSYZ20,Zeng19}.
Our protocol is useful whenever some  output states  are equivalent to  GHZ states under local Clifford transformations.

\subsection{Optimal verification of the $n$-qudit GHZ state}
Next, we generalize the above results to the qudit case, assuming that the local dimension $d$ is an odd prime. The qudit Pauli group is generated by  the phase operator $Z$ and the shift operator $X$ defined as follows,
\begin{equation}\label{eq:Z,X}
Z|j\>=\omega^j|j\>,\quad\  X|j\>=|j+1\>,\quad\  \omega=\rme^{2\pi\rmi/d},
\end{equation}
where $j\in \bbZ_d$ and $\bbZ_d$ is the ring of integers modulo $d$.

The concepts of admissible Pauli measurements/test operators  and canonical test projectors can be defined in a similar way as in Sec.~\ref{sec:QubitGHZ}. One admissible test is associated with the Pauli measurement
 $Z^n$ with the canonical test projector
\begin{equation}\label{eq:P0}
P_0=\sum_{j=0}^{d-1}(|j\>\<j|)^{\otimes n};
\end{equation}
the test is passed iff the outcomes of all $Z$ measurements coincide.
Each of the other admissible tests is associated with a string $\bfr\in \bbZ_d^n$ with $\sum_k r_k=0\mod d$, which means party $k$  performs the measurement on the eigenbasis of $XZ^{r_k}$ for $k=1,2,\dots,n$.
 The  canonical test projector reads
\begin{equation}\label{eq:Pother(d=p)}
P_{\bfr}=\frac{1}{d} \sum_{l=0}^{d-1} \biggl(\prod^{n}_{k=1} X_kZ_k^{r_k}\biggr)^l.
\end{equation}
Denote the outcome of party $k$  by an integer $o_k\in \bbZ_d$ corresponding to the eigenvalue $\omega^{o_k}$ of $XZ^{r_k}$; then the test is passed if $\sum_k o_k=0 \mod d$,
so that $\prod^{n}_{k=1} X_kZ_k^{r_k}$ has eigenvalue 1.
The following lemma is the analog of \lref{lem:TestProjAdmGHZ2} for the qudit case;
the proof is also similar and thus omitted.
\begin{lemma}\label{lem:TestProjGHZd}
Suppose $d$ is an odd prime.	Then the GHZ state $|\GHZ_n^d\>$ has $1+d^{n-1}$ admissible Pauli measurements and $1+d^{n-1}$ admissible test operators. Except for the test projector $P_0$ in \eref{eq:P0}, all other admissible test operators have the form in \eref{eq:Pother(d=p)} with $\sum_k r_k=0\mod d$.
\end{lemma}

Our verification protocol is composed of all $1+d^{n-1}$ admissible tests based on Pauli measurements. The test $P_0$ is performed with probability $1/(d+1)$ and the other $d^{n-1}$ tests are performed with  probability  $1/[(d+1)d^{n-2}]$ each.
The resulting verification operator is homogeneous and has the form
\begin{equation}\label{eq:Omega2}
\Omega_{\2} := \frac{1}{d+1}\biggl(P_0+\frac{1}{d^{n-2}}\sum_{\bfr}P_{\bfr}\biggr)=\frac{\openone+d|\GHZ_n^d\>\<\GHZ_n^d|}{d+1},
\end{equation}
where the second equality is proved in Appendix~\ref{app:ProofsEquations}.
We have $\beta(\Omega_{\2})=1/(d+1)$, and
\begin{equation}\label{eq:OptnuNQudit}
\nu(\Omega_{\2})=\frac{d}{d+1},\qquad N(\Omega_{\2})\approx \frac{d+1}{d\epsilon} \ln \delta^{-1}.
\end{equation}
Similar to the  qubit case, this protocol is optimal among all protocols based on separable measurements.
In addition, it is  essentially the unique optimal protocol based on Pauli measurements;  the number $1+d^{n-1}$ of  measurement settings is the smallest possible. Proposition~\ref{pro:GHZdunique} below
generalizes \pref{pro:GHZ2unique} to the qudit case. Its proof is a simple analog of the counterpart  for the qubit case and is thus omitted.
As in the qubit case, there is some freedom in choosing the Pauli group, and here we focus on the canonical Pauli group generated by the operators $Z$ and $X$ defined in  \eref{eq:Z,X} for each qudit.

\begin{proposition}\label{pro:GHZdunique}
	Suppose  $\Omega$ is an optimal  verification strategy with $\nu(\Omega)=d/(d+1)$ for $|\GHZ_n^d\>$ that is based on Pauli measurements, where $d$ is an odd prime. Then $\Omega=\Omega_{\2}$; in addition,
	$\Omega$ is composed of admissible tests with the same probabilities  as in 	 $\Omega_{\2}$.
\end{proposition}

\subsection{Alternative optimal protocol based on 2-designs}
When the local dimension $d$ is not necessarily a prime, we can still devise optimal protocols for verifying GHZ states by virtue of  (weighted complex projective) 2-designs \cite{Renes04, Scott06, RoyS07}.
Let $\{\caB_h\}^{m}_{h=0}$ be $m+1$ bases on the Hilbert space of dimension $d$, where $\caB_0$ is the standard basis, and each basis $\caB_h$ for $h=1,2,\ldots,m$ is composed of $d$ kets of the form
\begin{equation}\label{eq:2designGood}
|\psi_{h t}\>=\frac{1}{\sqrt{d}} \sum_{j=0}^{d-1} \rme^{\rmi \theta_{h t j}} |j\>,\quad\
\theta_{h t j}=2\pi \biggl[\frac{tj}{d}+\frac{h\tbinom{j}{2}}{m}\biggr]
\end{equation}
for $t\in \bbZ_d$.
Let  $w_0=1/(d+1)$ and $w_h=d/[m(d+1)]$ for  $h=1,2,\ldots,m$, and let
$\{\caB_h, w_h\}^{m}_{h=0}$ be a weighted set of kets with weight $w_h$ for all kets in basis $h$. When $d\geq 3$ and $m\geq\lceil\frac{3}{4}(d-1)^2\rceil$, the set $\{\caB_h, w_h\}^{m}_{h=0}$ forms a 2-design according to \rcite{RoyS07}.
Define
\begin{equation}
W:=\diag\big(\mu^0,\mu^1,\dots,\mu^{d-2},\mu^{-(d-1)(d-2)/2}\big),
\end{equation}
where $\mu=\rme^{2\pi\rmi/m}$ is a primitive $m$th root of unity.
Then $|\psi_{ht}\>$ is an eigenstate of $XW^h$ with eigenvalue $\omega^{-t}$ as shown in Appendix~\ref{app:ProofOmega3}, that is,
\begin{equation}\label{eq:XW^h}
XW^h=\sum_{t\in \bbZ_d} \omega^{-t}|\psi_{ht}\>\<\psi_{ht}|.
\end{equation}

When $d\geq 3$, by virtue of the 2-design $\{\caB_h, w_h\}^{m}_{h=0}$ we can construct an optimal protocol using  $1+m^{n-1}$ distinct tests.
The first test is still the standard test $P_0$ as given in \eref{eq:P0}.
Each of the other tests is specified by a string $\bfh\in \{1,2,\ldots,m\}^n$ with $\sum_k h_k=0\mod m$, which means
party $k$ (for $k=1,\dots,n$) performs the projective measurement on the basis $\caB_{h_k}$. The  outcome of party $k$ is denoted by $o_k\in \bbZ_d$, which corresponds to the ket $|\psi_{h_ko_k}\>$ and the eigenvalue $\omega^{-o_k}$ of $XW^h$.
The test is passed if $\sum_{k} o_k =0\mod d$, and  the test projector reads
\begin{equation}\label{eq:Pother}
P_{\bfh}=\frac{1}{d}\sum_{l=0}^{d-1}\biggl(\prod^{n}_{k=1} X_kW_k^{h_k}\biggr)^l.
\end{equation}
Note that all eigenvalues of $\prod^{n}_{k=1} X_kW_k^{h_k}$ are powers of $\omega$ according to \eref{eq:XW^h}, so $P_{\bfh}$ is the projector onto the eigenspace with eigenvalue 1. In addition, the target state $|\GHZ_n^d\>$ is stabilized by $\prod_{k=1}^n X_kW_k^{h_k}$ given the assumption  $\sum_k h_k=0\mod m$ and so can pass the test with certainty as desired.

We perform the test $P_0$ with probability $1/(d+1)$ and the other $m^{n-1}$ tests with probability $d/[(d+1)m^{n-1}]$ each.
The verification operator reads [cf.~$\Omega_{\2}$ in \eref{eq:Omega2}]
\begin{equation}\label{eq:Omega3}
\Omega_{\3}:=\frac{1}{d+1}\biggl(P_0+\frac{d}{m^{n-1}}\sum_{\bfh}P_{\bfh}\biggr)=\frac{\openone+d|\GHZ_n^d\>\<\GHZ_n^d|}{d+1},
\end{equation}
where the second equality is proved in Appendix~\ref{app:ProofOmega3}. This protocol is optimal among all protocols based on separable measurements. Compared with the protocol based on Pauli measurements, this protocol applies to GHZ states of any local dimension $d$ with $d\geq 3$, although it
requires more measurement settings. In addition, this protocol is the only homogeneous protocol for general GHZ states beyond qubit systems.

\begin{figure}
\begin{center}
	\includegraphics[width=7.5cm]{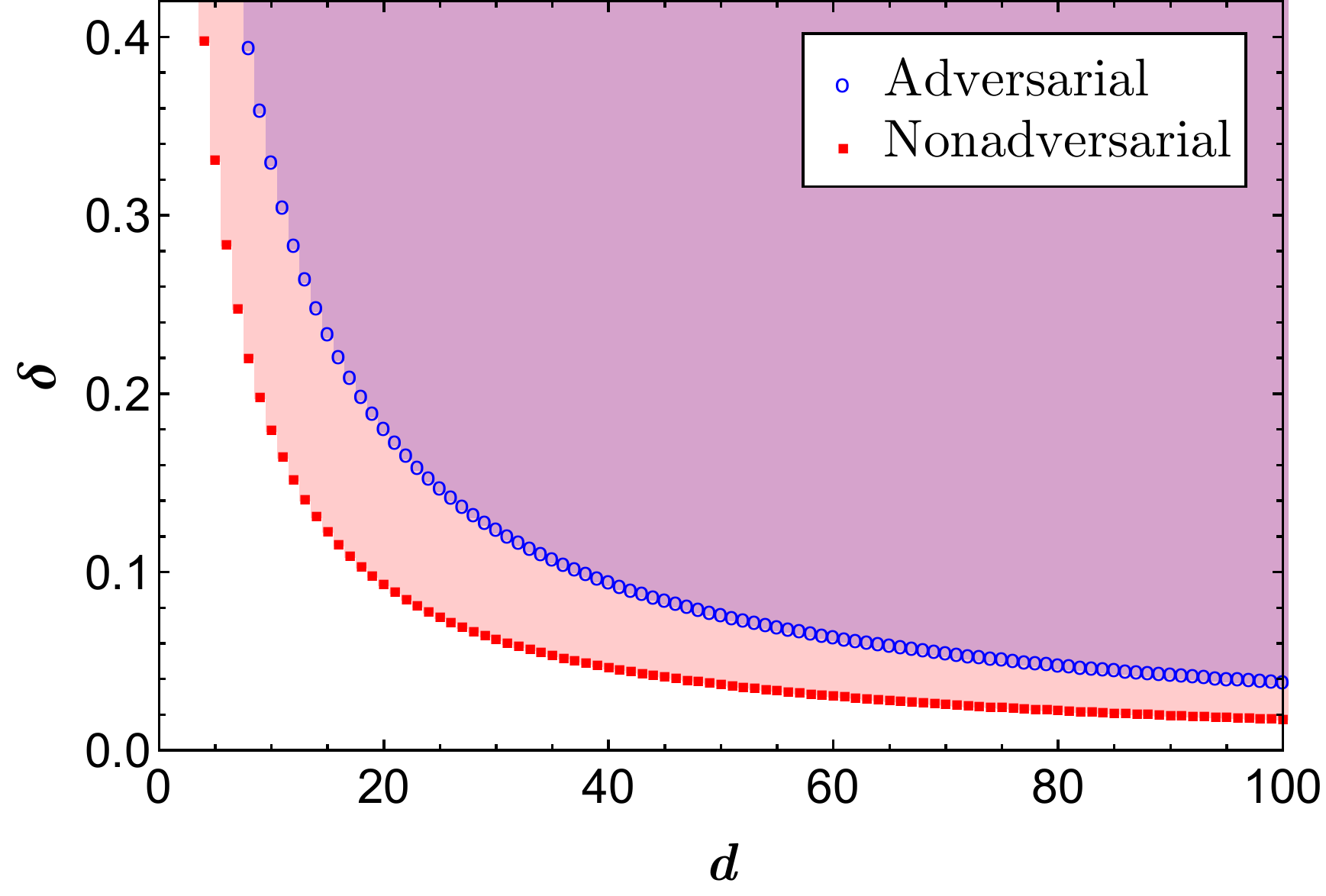}
	\caption{\label{fig:Sig} Certification of the GME of the $n$-qudit GHZ state in the adversarial scenario and the nonadversarial scenario using only one test. Here $d$ is the local dimension;  the  significance level
$\delta$ associated with the shaded region is achievable.  The homogeneous strategy $\Omega$ with $\beta(\Omega)=1/(d+1)$ [$\beta(\Omega)=2/(d+1)$] is applied to the nonadversarial scenario (adversarial scenario).
}
\end{center}
\end{figure}

\subsection{\label{sec:GME}Efficient certification of GME}
A quantum state $\rho$ is genuinely multipartite entangled (i.e., GME) if its fidelity with the GHZ state $\tr(\rho|\GHZ_n^d\>\<\GHZ_n^d|)$ is larger than $1/d$ \cite{Guhne09}.
To certify the GME of the qudit GHZ state with significance level $\delta$ using a given verification strategy $\Omega$, the number of tests is determined by \eref{eq:NumberTest} with $\epsilon=(d-1)/d$.
 If, in addition, $\Omega$ is the optimal local strategy with $\nu(\Omega)=d/(d+1)$, then this number reads
\begin{equation}
N_\mathrm{E}=\biggl\lceil\frac{ \ln \delta}{\ln2-\ln(d+1)}\biggr\rceil.
\end{equation}
We have  $N_\mathrm{E}=1$ when $d\geq2\delta^{-1}-1$, so the GME of the GHZ state
can be certified with any given significance level using only one test when the local dimension $d$ is sufficiently large, as illustrated in Fig.~\ref{fig:Sig}.
Compared with previous approaches for detecting GME that are based on witness operators \cite{TothG05,Zhou19},
our approach requires much fewer measurements.
Although single-copy entanglement detection is known before \cite{Dimic18,ZhuH19O}, single-copy detection of GME is still quite surprising, because it is much more difficult to demonstrate GME than just entanglement.

In sharp contrast, the previous verification protocols proposed in \rscite{PLM18,ZhuH19E} cannot certify the GME of GHZ states using a single test whenever $\delta\leq 1/2$ (the parameter range of practical interest).
To be specific,  the strategy $\Omega_{\rm PLM}$ in   \rcite{PLM18} only applies to the qubit case and has a spectral gap
 $\nu=2^{n-1}/(2^n-1)$. To certify the GME with significance level $\delta$, the number of tests required reads
\begin{equation}
N_\mathrm{E}(\Omega_{\rm PLM})=\biggl\lceil\frac{ \ln \delta}{\ln[1-2^{n-2}/(2^n-1)]}\biggr\rceil,
\end{equation}
so the GME cannot be certified using a single test when $\delta<5/7$ (for $n\geq 3$). The strategy $\Omega_{\rm ZH}$ in \rcite{ZhuH19E} applies to the qudit case and  has spectral gap $\nu=1/2$. The number of tests required reads
\begin{equation}
N_\mathrm{E}(\Omega_{\rm ZH})=\biggl\lceil\frac{ \ln \delta}{\ln(d+1)-\ln(2d)}\biggr\rceil,
\end{equation}
so  the GME cannot be certified using a single test when  $\delta\leq1/2$, irrespective of the local dimension $d$.
For example, to certify the GME of the GHZ state  with significance level $\delta=0.01$ (0.001),  the strategy $\Omega_{\rm PLM}$  in   \rcite{PLM18} requires  at least 14 (21) tests, while the strategy $\Omega_{\rm ZH}$ in \rcite{ZhuH19E}  requires  at least
 7 (10) tests. These observations demonstrate that our protocols are much more efficient than previous protocols for certifying GME.

\section{Verification of GHZ-like states}
Next, consider the GHZ-like states
\begin{equation}\label{eq:GHZlike}
|\xi\>=\sum_{j=0}^{d-1}\lambda_j|j\>^{\otimes n},
\end{equation}
where the coefficients $\lambda_j$
have decreasing order $1> \lambda_0\geq\lambda_1\geq \cdots \lambda_{d-1}\geq 0$ and satisfy  $\sum_{j=0}^{d-1}\lambda_j^2 = 1$.
Such states are of interest to quantum state sharing \cite{Gordon06} and foundational  studies on nonlocality \cite{Zukow02,Cerece04}.
They are also useful in improving signal-to-noise ratios in interferometry \cite{Leib04}
and enhancing signal amplitudes of the electronic spin readout \cite{Jiang09}.

\subsection{Simplest protocol for verifying  GHZ-like states}
We first show that the GHZ-like state $|\xi\>$ can be verified efficiently
using only two distinct tests constructed from mutually unbiased bases (MUB).
Recall that two bases $\{|\psi_i\>\}_{i=0}^{d-1}$ and $\{|\varphi_j\>\}_{j=0}^{d-1}$  for a Hilbert space of dimension $d$
are mutually unbiased if they satisfy $|\<\psi_i|\varphi_j\>|^2=1/d$ for all $i$ and $j$ \cite{Ivano81, Woot89, Durt10}.
Let $\caB_0$ be the standard computational basis and let  $\caB=\{\ket{u_g}\}_{g\in \bbZ_d}$ be any basis that is unbiased with  $\caB_0$. A simple example of $\caB$ is the Fourier basis $\bigl\{\sum_{j=0}^{d-1} \omega^{g j} |j\>/\sqrt{d}\bigr\}_{g\in \bbZ_d}$ with $\omega=\rme^{ 2\pi \rmi/d}$, which happens to be the eigenbasis of the shift operator $X$ in \eref{eq:Z,X}. The following discussion is independent of the choice of the basis $\caB$ as long as it is unbiased with respect to the standard basis $\caB_0$.

The first test is the standard test $P_0$ in \eref{eq:P0}. For the second test, the first $n-1$ parties perform projective measurements on the basis $\caB$. If they obtain the outcome $\bfg=\{g_1,g_2,\dots,g_{n-1}\}\in \bbZ_d^{n-1}$, then the normalized reduced state of party $n$ reads
\begin{equation}
d^{\frac{n-1}{2}}\bigg(\bigotimes^{n-1}_{k=1}\<u_{g_k}|\bigg)|\xi\>=M|v_{\bfg}\>,
\end{equation}
where
\begin{align}
|v_{\bfg}\>&:=d^{\frac{n-1}{2}}\bigg(\bigotimes^{n-1}_{k=1}\<u_{g_k}|\bigg)|\GHZ_n^d\>, \\
M&:=\sqrt{d} \diag(\lambda_0,\lambda_1,\dots,\lambda_{d-1}). \label{eq:M}
\end{align}
Note that $|v_{\bfg}\>$ has a constant overlap of $1/d$ with each element in the basis $\caB_0$.
Then  party $n$ performs the projective measurement $\{M|v_{\bfg}\>\<v_{\bfg}|M,I-M|v_{\bfg}\>\<v_{\bfg}|M\}$,
where $I$ is the identity operator on the Hilbert space of one qudit.
The test is passed if party $n$ obtains the first outcome (corresponding to $M|v_{\bfg}\>\<v_{\bfg}|M$).  The resulting test projector reads
\begin{equation}
P_1=\sum_{\bfg}\left[\bigotimes^{n-1}_{k=1}(|u_{g_k}\>\<u_{g_k}|)\right]\otimes \big(M|v_{\bfg}\>\<v_{\bfg}|M\big).
\end{equation}
So we have
\begin{equation}
\tr(P_0 P_1)=\frac{1}{d^{n-1}}\sum_{\bfg}\sum_{j=0}^{d-1}|\<j|M|v_{\bfg}\>|^2
=\frac{1}{d^{n-1}}\sum_{\bfg,\,j}\lambda_j^2=1,
\end{equation}
which implies that the two projectors $\bar{P}_0:=P_0-|\xi\>\<\xi|$ and $\bar{P}_1:=P_1-|\xi\>\<\xi|$ have orthogonal supports.

If we perform the two tests $P_0$ and $P_1$ with probability $p$ and $1-p$, respectively, then the
verification operator reads $\Omega_{\4}=pP_0+(1-p)P_1$, with
\begin{align}
\beta(\Omega_{\4})&=\bigl\|\bar{\Omega}_{\4}\bigr\|=\max\{p, 1-p\}\geq\frac{1}{2},
\end{align}
where $\bar{\Omega}_{\4}=\Omega_{\4}-|\xi\>\<\xi|$.
The lower bound is saturated iff $p=1/2$, in which case we have $\Omega_{\4}=(P_0+P_1)/2$.
The corresponding spectral gap $\nu(\Omega_{\4})$ and the number $N(\Omega_{\4})$ of required tests read
\begin{equation}
\nu(\Omega_{\4})=\frac{1}{2},\qquad N(\Omega_{\4})\approx \frac{2}{\epsilon} \ln \delta^{-1}.
\end{equation}
According to \rcite{ZhuH19O}, here the spectral gap attains the maximum among all protocols composed of two distinct local projective tests, so the above protocol is the most efficient among all protocols based on two distinct local projective tests.

\subsection{\label{sec:GHZlikeOpt}Optimal protocol under one-way LOCC}
For a bipartite state $|\zeta\>=\sum_{j=0}^{d-1} \lambda_j |jj\>$ with the same local dimension and coefficients $\lambda_j$ as $|\xi\>$ in  \eref{eq:GHZlike}, the maximum spectral gap of any verification operator based on one-way LOCC is $1/(1+\lambda_0^2)$ \cite{Wang19,Yu19}. The counterpart for the  GHZ-like state $|\xi\>$ cannot be larger. Here we shall demonstrate that this upper bound can be saturated.
When $d\geq3$, our protocol consists of $1+m^{n-1}$ distinct tests with $m\geq\lceil\frac{3}{4}(d-1)^2\rceil$.
The first one is  the standard test in \eref{eq:P0}.
For each of the other tests, the first $n-1$ parties perform projective measurements on the bases $\caB_{h_1}, \caB_{h_2}, \dots, \caB_{h_{n-1}}$ [cf.~\eref{eq:2designGood}], respectively, where $h_1, h_2, \ldots, h_{n-1}\in \{1,2,\ldots,m\}$. After receiving the outcomes $o_1,o_2,\dots,o_{n-1}\in \bbZ_d$ of these measurements, we choose $h_n,o_n$ to satisfy the conditions $\sum_{k=1}^n h_k=0\mod m$ and $\sum_{k=1}^n  o_k=0\mod d$. Then  party $n$ performs the projective measurement $\{MP_{h_n o_n}M,I-MP_{h_n o_n}M\}$, where $P_{h_k o_k}=|\psi_{h_k o_k}\>\<\psi_{h_k o_k}|$ and $M$ is defined in \eref{eq:M}.
The test is passed if party $n$ obtains the first outcome (corresponding to $MP_{h_n o_n}M$), and the  test projector reads
\begin{equation}\label{eq:Pother'}
P'_{\bfh}=\big(I^{\otimes (n-1)}\otimes M\big)P_{\bfh}\big(I^{\otimes (n-1)}\otimes M\big),
\end{equation}
where $P_{\bfh}$ is the test projector  in \eref{eq:Pother}.

Suppose we perform the test $P_0$ with probability $p$ and each of the other tests with probability $(1-p)/m^{n-1}$;
then the verification operator reads
\begin{equation}\label{eq:Omega4}
\Omega_{\5}=p P_0+ (1-p)\Pi,
\end{equation}
where
\begin{equation}\label{eq:Pi1}
\Pi:=\frac{\sum_{\bfh}P'_{\bfh}}{{m^{n-1}}}
=|\xi\>\<\xi|+I^{\otimes (n-1)}\otimes \rho_{n}-\sum_{j=0}^{d-1}\lambda_j^2(|j\>\<j|)^{\otimes n},
\end{equation}
with $\rho_{n}=\tr_{1,2,\dots,n-1}(|\xi\>\<\xi|)=\sum_{j=0}^{d-1} \lambda_j^2 |j\>\<j|$ being the reduced state for party $n$. Here the second equality follows from Eqs.~\eqref{eq:Omega3}  and \eqref{eq:Pother'} [cf.~\eref{eq:sumPh} in Appendix~\ref{app:ProofOmega3}].
Note that $\bar{\Pi}=\Pi-|\xi\>\<\xi|$ and $\bar{P}_0=P_0-|\xi\>\<\xi|$ are orthogonal; we conclude that
\begin{align}
\beta(\Omega_{\5})&=\bigl\|\bar{\Omega}_{\5}\bigr\|=\max\{p,(1-p)\|\bar{\Pi}\|\}\nonumber\\
&=\max\{p,(1-p)\lambda_0^2\}\geq\frac{\lambda_0^2}{1+\lambda_0^2}.\label{eq:betaOmega4}
\end{align}
The  bound is saturated iff $p={\lambda_0^2}/(1+\lambda_0^2)$, in  which case we have
\begin{equation}\label{eq:nuOmega5}
\nu(\Omega_{\5})=\frac{1}{1+\lambda_0^2},\qquad N(\Omega_{\5})\approx \frac{1+\lambda_0^2}{\epsilon} \ln \delta^{-1}.
\end{equation}
Therefore, this protocol is optimal among all protocols based on one-way LOCC.

When the local dimension $d$ is a prime, the number of distinct tests required for constructing the optimal protocol can be reduced to $1+d^{n-1}$.
Take the qubit case for example. The first test is still the standard test $P_0$. For each of the other tests, the first $n-1$ parties perform either $X$ or $Y$ measurements. Then party $n$ performs the projective measurement $\{|v\>\<v|,I-|v\>\<v|\}$, where $|v\>$ is the normalized reduced state of party $n$ depending on the outcomes of the first $n-1$ parties. The test is passed if party $n$ obtains  the first outcome (corresponding to $|v\>\<v|$). The test projector has the form
\begin{equation}\label{eq:PY'}
P'_{\scrY}=\big(I^{\otimes (n-1)}\otimes M\big)P_{\scrY}\big(I^{\otimes (n-1)}\otimes M\big),
\end{equation}
where $\scrY\subset \{1,2,\ldots, n\}$ has even cardinality  and $P_{\scrY}$ is the test projector  in \eref{eq:PY}. Suppose we perform the test $P_0$ with probability $p$ and each of the other tests with probability $(1-p)/2^{n-1}$;
 then the verification operator reads
\begin{equation}
\Omega'_{\5}=p P_0+ \frac{1-p}{{2^{n-1}}}\sum_{\scrY}P'_{\scrY}.
\end{equation}
Again,  the maximum spectral gap $\nu(\Omega'_{\5})=1/(1+\lambda_0^2)$ is attained when $p={\lambda_0^2}/(1+\lambda_0^2)$. When $d$ is an odd prime,  more details can be found
in Appendix~\ref{app:StratGHZ-like}.

\subsection{\label{sec:GHZlikeTwoWay}Improved protocol based on more communications}
The above protocol for verifying GHZ-like states can be improved further if more communications are allowed.
Let $\Omega_k$ ($k=1,2,\dots, n$) be the strategy defined according to \eref{eq:Omega4}, but with the roles of party $k$ and party $n$ interchanged; that is, the measurement performed by party $k$ depends on the measurement outcomes of the other $n-1$ parties.
Then we can construct a new strategy by applying  $\Omega_1,\Omega_2,\dots,\Omega_n$ with probability $1/n$ each, and the resulting verification operator reads
\begin{equation}\label{eq:Omega5}
\Omega_{\6}=\frac{1}{n}\sum_{k=1}^n\Omega_k=p P_0+ (1-p)\frac{1}{n}\sum_{k=1}^n\Pi_k.
\end{equation}
Here the operator  $\Pi_k$ is derived from $\Pi$ in \eref{eq:Pi1}  by replacing $I^{\otimes (n-1)}\otimes \rho_{n}$ with $R_k:=I^{\otimes (k-1)}\otimes \rho_{k}\otimes I^{\otimes (n-k)}$, where  $\rho_{k}=\sum_{j=0}^{d-1} \lambda_j^2 |j\>\<j|$. We have
\begin{gather}
\frac{1}{n}\sum_{k=1}^n\Pi_k
=|\xi\>\<\xi|+\frac{1}{n}\sum_{k=1}^n R_k-\sum_{j=0}^{d-1}\lambda_j^2(|j\>\<j|)^{\otimes n}, \label{eq:Gamma}\\
\biggl\|\frac{1}{n}\sum_{k=1}^n\Pi_k
-|\xi\>\<\xi|\biggr\|=\frac{(n-1)\lambda_0^2+\lambda_1^2}{n}.
\end{gather}
Therefore,
\begin{align}
\beta(\Omega_{\6})&=\max\bigl\{p,(1-p)n^{-1}[(n-1)\lambda_0^2+\lambda_1^2]\bigr\} \nonumber\\
&\geq [n+(n-1)\lambda_0^2+\lambda_1^2]^{-1}[(n-1)\lambda_0^2+\lambda_1^2].
\end{align}
The bound is saturated when
\begin{equation}
p=\frac{(n-1)\lambda_0^2+\lambda_1^2}{n+(n-1)\lambda_0^2+\lambda_1^2},
\end{equation}
in which case we have
\begin{equation}
\nu(\Omega_{\6})=\frac{n}{n+(n-1)\lambda_0^2+\lambda_1^2}\geq\nu(\Omega_{\5}).
\end{equation}
The strategy $\Omega_{\6}$ is more efficient than $\Omega_{\5}$ except when $\lambda_1=\lambda_0$, as illustrated in Fig.~\ref{fig:NumT}.

\section{Adversarial scenario}
Finally, we turn to the adversarial scenario, in which the quantum device is controlled by a potentially malicious adversary,
and can produce an arbitrary correlated or entangled state $\rho$ on the whole system $\caH^{\otimes (N+1)}$  \cite{HayaM15,TakeM18}.
By virtue of a general recipe proposed in \rscite{ZhuEVQPSshort19,ZhuEVQPSlong19}, we can verify the target state $|\Psi\>$ in the adversarial scenario by first randomly choosing $N$ systems and then applying a verification strategy $\Omega$ to each system chosen. Note that only one-way communication from the adversary to the verifier is involved. In addition,
usually the choices of the $N$ systems and the specific test for each system chosen can be determined after receiving the state $\rho$. Therefore, the adversary has no information about these choices before sending the state $\rho$, and he/she cannot get any advantage even if these choices are broadcast after sending the state $\rho$. By constructing a suitable strategy $\Omega$, we can make sure with high confidence (low significance level) that the reduced state on the remaining system has fidelity at least $1-\epsilon$ if all $N$ tests are passed.
Efficient state verification in such an adversarial scenario is crucial to quantum secret sharing \cite{Titt01,Chen05} and quantum networks \cite{McCut16}.

If there is no restriction on the accessible measurements, then the optimal strategy can be chosen to be homogeneous \cite{ZhuEVQPSshort19,ZhuEVQPSlong19}.
In the high-precision limit  $\epsilon,\delta\rightarrow 0$,  the minimal number of tests required to verify $|\Psi\>$ within infidelity $\epsilon$ and significance level $\delta$ reads   \cite{ZhuEVQPSshort19,ZhuEVQPSlong19} (assuming $\beta(\Omega)>0$)
\begin{equation}\label{eq:NumTestAdv}
N\approx[\beta(\Omega)\epsilon\ln \beta(\Omega)^{-1}]^{-1}  \ln \delta^{-1}.
\end{equation}
This number is minimized when $\beta(\Omega)=1/\rme$, which yields $N\approx \rme\epsilon^{-1} \ln \delta^{-1}$. In addition, this number increases monotonically when $\beta(\Omega)$ deviates from the value $1/\rme$. If $\epsilon,\delta$ are  small but not infinitesimal, say $\epsilon,\delta \leq 0.01$, then the choice
$\beta(\Omega)=1/\rme$ is nearly optimal even if it is not exactly optimal.

Our strategies for verifying the qudit GHZ state $|\GHZ_n^d\>$ are  homogeneous with $\beta(\Omega)=1/(d+1)$.
To construct the optimal verification strategy in the adversarial scenario,
it suffices to add the trivial test with a suitable probability $p$.
The  test operator associated with the trivial test is the identity operator, so all states can pass the test for sure. Let $p=[(d+1)\beta-1]/d$ with $1/(d+1)\leq \beta<1$; then the verification operator reads
\begin{align}
\Omega_{\7}:=&(1-p)\frac{\openone+ d|\GHZ_n^d\>\<\GHZ_n^d|}{d+1}+p\openone \nonumber\\
=&|\GHZ_n^d\>\<\GHZ_n^d|+\beta(\openone-|\GHZ_n^d\>\<\GHZ_n^d|).
\end{align}
Any homogeneous strategy $\Omega$ with $1/(d+1)\leq\beta(\Omega)<1$ can be so constructed using local projective measurements.
In particular, by choosing $p=(d+1-\rme)/(\rme d)$, we can construct the  homogeneous strategy $\Omega_{\7}$ with $\beta(\Omega_{\7})=1/\rme$, which is optimal for high-precision verification in the adversarial scenario (the optimal value may be slightly different when $\epsilon, \delta$ are  small but not infinitesimal). Similarly, we can construct a  homogeneous strategy $\Omega$ with $\beta(\Omega)=2/(d+1)$, with which the GME can be certified in the adversarial scenario using only one test as long as the significance level satisfies $\delta\geq 4d/(d+1)^2$, as illustrated in Fig.~\ref{fig:Sig}. This claim follows from Corollary~6 in \rcite{ZhuEVQPSlong19} with $\epsilon=(d-1)/d$
(see also Theorem~3 in \rcite{ZhuH19O}).
In sharp contrast, previous protocols in \rscite{PLM18, ZhuH19E} cannot certify the GME using a single test
whenever $\delta\leq1/2$ even in the nonadversarial scenario (cf.~Sec.~\ref{sec:GME}), not to mention the adversarial scenario.

\begin{figure}
\begin{center}
	\includegraphics[width=7.5cm]{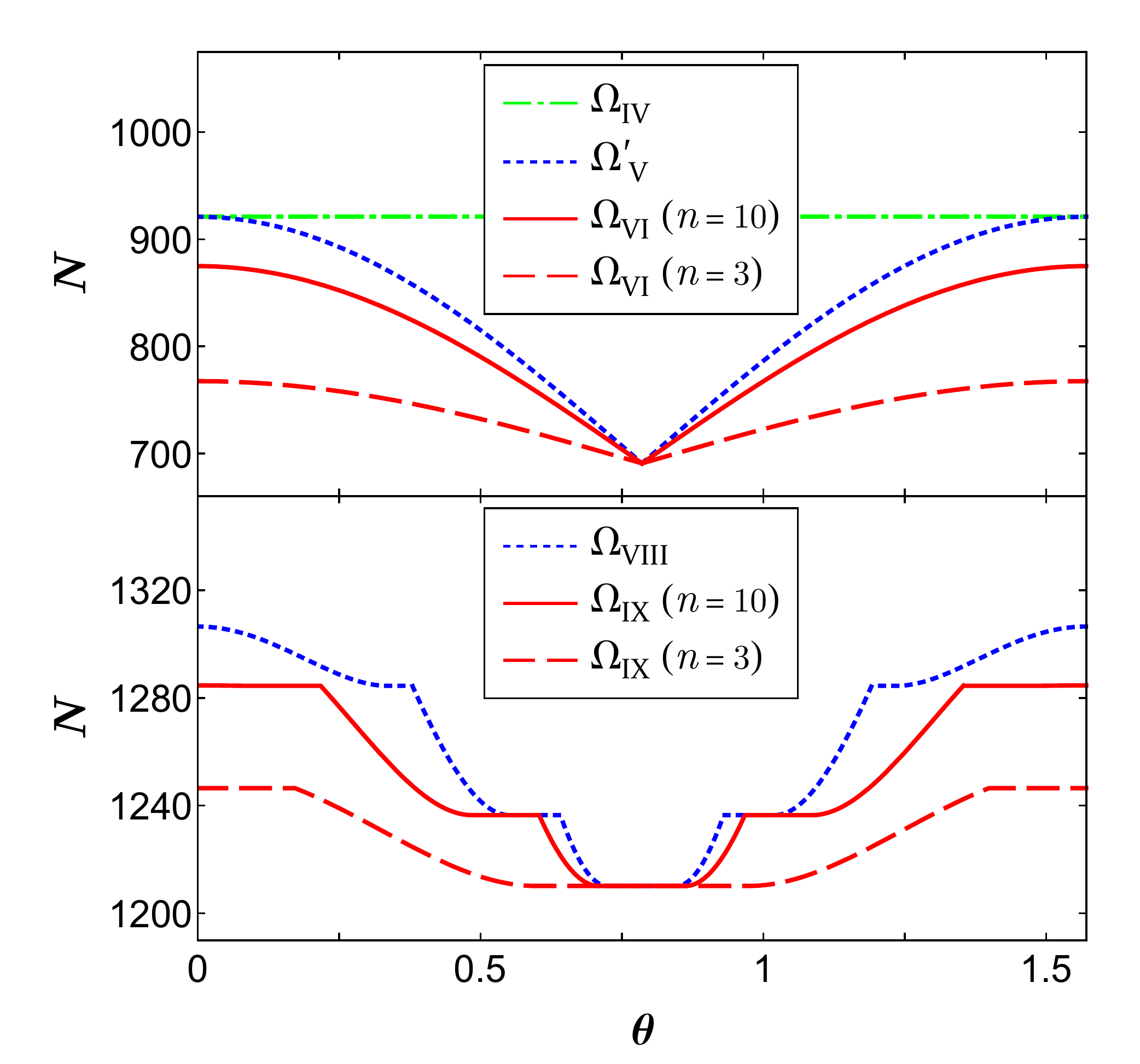}
	\caption{\label{fig:NumT} Efficient verification of the $n$-qubit GHZ-like state $|\xi\>=\cos\theta\ket{0}^{\otimes n}+\sin\theta\ket{1}^{\otimes n}$ in the nonadversarial scenario (upper plot) and the adversarial scenario (lower plot).
Here $N$ is the number of tests required to achieve infidelity $\epsilon=0.01$ and significance level $\delta=0.01$.
Note that $N(\Omega_{\6})$ and $N(\Omega_{\9})$ are dependent on the qubit number $n$, while $N(\Omega_{\4})$, $N(\Omega'_{\5})$, and $N(\Omega_{\8})$ are not.
}
\end{center}
\end{figure}

Next, we  devise a homogeneous strategy for verifying the GHZ-like state $|\xi\>$ in \eref{eq:GHZlike} by modifying the strategy $\Omega_{\5}$ in \eref{eq:Omega4}, which requires one-way communication. Let $\lambda_0^2/(1+\lambda_0^2)\leq p<1$ and  replace the test  projector $P_0$ in \eref{eq:P0} with the following test operator
\begin{equation}
Q_0=P_0+\sum_{\bfj\in \mathscr{B}}\Bigl[1-\Bigl(\frac{1}{p}-1\Bigr)\lambda_{j_n}^2\Bigr]|\bfj\>\<\bfj|,
\end{equation}
where $\mathscr{B}$ denotes the subset of $\bbZ_d^n$ excluding elements $\bfj$ that satisfy $j_1=j_2=\cdots= j_n$. Note that $Q_0$ can be realized by local projective measurements: All $n$ parties perform projective measurements on the standard basis; the test is passed with certainty if they obtain the same outcome, while with probability $1-(p^{-1}-1)\lambda_{j_n}^2$
if they do not obtain the same outcome and party $n$ obtains outcome $j_n$.
Then the verification operator $\Omega_{\5}$ turns into
\begin{equation}
\Omega_{\8}=pQ_0+(1-p)\Pi=|\xi\>\<\xi|+p(\openone-|\xi\>\<\xi|),
\end{equation}
which is homogeneous with $\beta(\Omega)=p$. Here $\Pi$ is defined in \eref{eq:Pi1}. To achieve optimal performance in
high-precision verification in the adversarial scenario, we can choose $p=\max\{\rme^{-1},\lambda_0^2/(1+\lambda_0^2)\}$.  If $\lambda_0^2\leq 1/(\rme-1)$, then  we have $\beta(\Omega)=1/\rme$, so the  homogeneous strategy $\Omega_{\8}$ constructed in this way is optimal even  among strategies that can access entangling measurements. In general, $\Omega_{\8}$ is optimal among all strategies based on one-way LOCC.
Even in the worst case with $\beta(\Omega)=1/2$,  the number of required tests is only
$2(\ln \delta^{-1})/(\epsilon\ln2)$, and the overhead compared with the optimal strategy based on entangling measurements is only about
$6\%$. By contrast, the choice  $p=\lambda_0^2/(1+\lambda_0^2)$ is  optimal for  fidelity estimation.

The strategy $\Omega_{\6}$ in \eref{eq:Omega5} can also be turned into a homogeneous strategy.
Let
\begin{equation}\label{eq:pBoundTwoway}
\frac{(n-1)\lambda_0^2+\lambda_1^2}{n+(n-1)\lambda_0^2+\lambda_1^2}\leq p<1
\end{equation}
and replace the projector $P_0$ with the following operator
\begin{equation}
\tilde{Q}_0=P_0+\sum_{\bfj\in \mathscr{B}}\biggl[1-\frac{1}{n}\Bigl(\frac{1}{p}-1\Bigr)\sum_{k=1}^{n}\lambda_{j_k}^2\biggr]|\bfj\>\<\bfj|,
\end{equation}
which can  be realized by  local projective measurements in analogy to  $Q_0$.
The resulting verification operator  reads
\begin{equation}
\Omega_{\9}=p \tilde{Q}_0+(1-p)\frac{1}{n}\sum_{k=1}^n\Pi_k=|\xi\>\<\xi|+p(\openone-|\xi\>\<\xi|),
\end{equation}
which  is homogeneous with $\beta(\Omega)=p$.
For high-precision verification in the adversarial scenario, the optimal choice of $p$ is
\begin{equation}
p=\max\biggl\{\rme^{-1},\frac{(n-1)\lambda_0^2+\lambda_1^2}{n+(n-1)\lambda_0^2+\lambda_1^2}\biggr\}.
\end{equation}
The resulting  strategy $\Omega_{\9}$ is optimal if
\begin{equation}
(n-1)\lambda_0^2+\lambda_1^2\leq \frac{n}{\rme-1},
\end{equation}
in which case we have $\beta(\Omega)=1/\rme$.
For fidelity estimation,  the optimal choice of $p$ coincides with the lower bound in \eref{eq:pBoundTwoway}, that is,
\begin{equation}
\frac{(n-1)\lambda_0^2+\lambda_1^2}{n+(n-1)\lambda_0^2+\lambda_1^2}.
\end{equation}

Although a lot of random bits are required to construct the above verification protocols,
they can be generated by classical random number generators. In the adversarial scenario  we consider, the adversary  controls the preparation of quantum states, but not the measurement devices used to verify these states, which is in sharp contrast with self-testing \cite{ZhuEVQPSlong19,Mayers04, Supic19}. In addition, there is  only one-way communication from the adversary to the verifier.
Usually the choices of the systems for verification/application and the specific test for each system chosen can be determined after the state is received from the adversary as mentioned in the beginning of this section. Therefore, even pseudo-random number generators like computer programs are sufficient for this task.
Quantum resources are not necessary to generate random bits.

\section{Summary}
We proposed optimal and homogeneous strategies for verifying GHZ states based on local projective measurements. Only Pauli measurements are required when the local dimension is a prime.  These protocols are also surprisingly efficient for
estimating the fidelity and certifying the GME.  In particular, they enable the certification  of the GME with any given significance level using only one test when the local dimension is sufficiently large. Such a high efficiency has never been achieved or even anticipated before. Our results indicate that it is easier to certify GME than thought previously.
We hope that these results will be demonstrated in experiments in the near future.
Moreover, our protocols can be generalized to verify GHZ-like states and can be applied to the adversarial scenario, while retaining a high efficiency.
Our study provides an efficient tool for evaluating the qualities of GHZ states and GHZ-like states prepared  in the lab. Meanwhile, it offers valuable insights into the verification, fidelity estimation, and entanglement certification of multipartite quantum states. In addition, the concepts of admissible measurements/test operators and canonical test projectors we introduced are useful beyond the focus of this work. In the future it would be desirable to generalize our results to other important multipartite quantum states.

\section*{Acknowledgment}
HZ is grateful to Masahito Hayashi for discussions. This work is  supported by  the National Natural Science Foundation of China (Grant No. 11875110).

\appendix

\section{\label{app:ProofProposition}Proofs of \lref{lem:TestProjAdmGHZ2} and \pref{pro:GHZ2unique} }

\begin{proof}[Proof of Lemma~\ref{lem:TestProjAdmGHZ2}]
	To determine admissible Pauli measurements, we need to consider canonical test projectors associated with Pauli measurements. First note that an incomplete Pauli measurement cannot be admissible. To see this, it suffices to consider the case in which the incomplete Pauli measurement has weight $n-1$. After $n-1$ single-qubit Pauli measurements, the reduced states of the remaining party for all possible outcomes are eigenstates of one Pauli operator, so we can obtain a smaller canonical test projector by performing a suitable Pauli measurement on the remaining qubit. Therefore, it suffices to consider canonical test projectors associated with complete Pauli measurements.
	
	Denote by $\scrX, \scrY, \scrZ$ the sets of parties that perform $X, Y, Z$ measurements, respectively. If $|\scrZ|\geq 1$, then the canonical test projector reads
	\begin{equation}\label{eq:TestProjXYZ}
	\openone_{\scrX\cup \scrY}\otimes\Biggl[ \bigotimes_{j\in \scrZ} (|0\>\<0|)_j+\bigotimes_{j\in \scrZ} (|1\>\<1|)_j\Biggr]\geq P_0,
	\end{equation}
	where $\openone_{\scrX\cup \scrY}$ denotes the identity operator associated with parties in ${\scrX\cup \scrY}$, and the subscript $j$ specifies the party on which the operators $|0\>\<0|$ and $|1\>\<1|$ act. The inequality in Eq.~\eqref{eq:TestProjXYZ} is saturated iff all parties perform $Z$ measurements, that is, $|\scrZ|=n$. Note that every test projector of $|\GHZ_n^2\>$ that is based on a Pauli measurement has rank at least 2, so the test projector $P_0$ and the corresponding Pauli measurement are admissible, while other Pauli measurements with $|\scrZ|\geq 1$ and corresponding  test projectors are not admissible.
	
	Next, suppose  each party performs either $X$ or $Y$ measurement, that is, $|\scrZ|=0$ and $|\scrX|+|\scrY|=n$.    If $|\scrY|$ is odd, then the canonical test projector is equal to the identity and so cannot be admissible, given that all states in the measurement basis have nonzero overlaps with $|\GHZ_n^2\>$.  If $|\scrY|$ is even, then the canonical test projector is  $P_\scrY$ given in \eref{eq:PY}, which has rank $2^{n-1}$. Note that $P_\scrY\ngeq P_0$, and there is no other canonical test projector that is smaller than $P_\scrY$. Therefore, all test projectors  $P_\scrY$ with even $|\scrY|$ and corresponding Pauli measurements are admissible.
	
	In summary, there are  $1+2^{n-1}$ admissible Pauli measurements, namely, $Z^n$ and all strings in $\{X,Y\}^n$ with even numbers of $Y$. The corresponding $1+2^{n-1}$ canonical test projectors in  \eqsref{eq:TestProjZn}{eq:PY} are the only admissible test operators.
\end{proof}

\begin{proof}[Proof of \pref{pro:GHZ2unique}]
	To start with, suppose the verification strategy $\Omega$ only consists of admissible test projectors. According to Lemma~\ref{lem:TestProjAdmGHZ2}, $\Omega$ can be expressed as
	\begin{equation}
	\Omega=p_0 P_0 +\sum_\scrY p_\scrY P_\scrY, \quad p_0, p_\scrY\geq0,\; p_0+ \sum_\scrY p_\scrY=1.
	\end{equation}
Now the   assumption $\nu(\Omega)\!=\!2/3$ sets an upper bound for $p_0$, that is, $p_0\leq \beta(\Omega)\!=\!1/3$. Therefore,
	\begin{equation}\label{eq:OmegaTrace}
	\tr(\Omega)=2p_0+2^{n-1}\sum_\scrY p_\scrY=2p_0+2^{n-1}(1-p_0)\geq \frac{2^n+2}{3},
	\end{equation}
	where the inequality is saturated iff $p_0\!=\!1/3$.
	In addition,
	\begin{equation}
	\beta(\Omega)\geq \frac{\tr(\Omega)-1}{2^n-1}
	\geq \frac{2^n-1}{3(2^n-1)}= \frac{1}{3}.
	\end{equation}
	The first inequality is saturated iff $\Omega$ is homogeneous, which means all eigenvalues of $\Omega$ are equal except for the largest one. The second inequality is saturated iff the inequality in  \eref{eq:OmegaTrace} is saturated, which implies that $p_0=1/3$. If $\nu(\Omega)=2/3$, that is,  $\beta(\Omega)=1/3$, then both inequalities are saturated, so that
	\begin{align}\label{eq:OmegaDecom}
	\Omega&= \frac{1}{3}\big(\openone+ 2|\GHZ_n^2\>\<\GHZ_n^2|\big) \nonumber\\
	&=\Omega_{\1}=\frac{1}{3}\biggl(P_0+\frac{1}{2^{n-2}}\sum_{\scrY}P_{\scrY}\biggr).
	\end{align}
	Moreover, the decomposition in the last expression  is unique because the $1+2^{n-1}$ admissible canonical test  projectors are linearly independent in the operator space.
	
	If $\Omega$ consists of some nonadmissible test operators, we can construct a new strategy $\Omega'$ by replacing each nonadmissible test operator $E$ with an admissible test projector $P$ satisfying $P\leq E$ and $\tr(P)< \tr(E)$. Then $\Omega'\leq \Omega$ is an optimal strategy with $\nu(\Omega')=2/3$, which implies that $\Omega'=\Omega_{\1}$  according to the above discussion. In addition,  we have $\nu(\Omega)< \nu(\Omega')=2/3$ since $\Omega'$ is homogeneous and $\tr(\Omega')<\tr(\Omega)$.
	This conclusion contradicts the assumption that $\nu(\Omega)=2/3$, which completes the proof of \pref{pro:GHZ2unique}.
\end{proof}

Incidentally, Proposition~\ref{pro:GHZdunique} can be proved using a similar reasoning used in the proof of \pref{pro:GHZ2unique}. Accordingly, Lemma~\ref{lem:TestProjAdmGHZ2} featuring in the above proof  can be replaced by \lref{lem:TestProjGHZd}, which applies to the qudit case, assuming that $d$ is an odd prime.

\section{\label{app:ProofsEquations}Proofs of \eqsref{eq:Omega1}{eq:Omega2}}
\begin{proof}[Proof of \eref{eq:Omega1}]
Note that the sum of all test projectors $P_{\scrY}$ with $\scrY\subset \{1,\dots,n\} $ of even cardinalities can be expressed as
\begin{align}\label{eq:sumPj}
\sum_{\scrY}P_{\scrY}\!
&= 2^{n-2}\openone+ \frac{1}{2}\!\sum_{t=0}^{\lfloor n/2 \rfloor} \!(-1)^t \sum_{j}\caP_j\bigl\{Y^{\otimes 2t}\!\otimes X^{\otimes (n-2t)}\bigr\}\nonumber\\
&= 2^{n-2}\openone+ \frac{1}{4}\big[(X+ \rmi Y)^{\otimes n}+(X- \rmi Y)^{\otimes n} \big]\nonumber\\
&= 2^{n-2}\big[\openone+(|0\>\<1|)^{\otimes n}+(|1\>\<0|)^{\otimes n}\big],
\end{align}
where $\sum_{j}\caP_j\{Y^{\otimes 2t}\otimes X^{\otimes (n-2t)}\}$ denotes the sum over $\binom{n}{2t}$
distinct  permutations of $Y^{\otimes 2t}\otimes X^{\otimes (n-2t)}$. This equation implies the second equality in \eref{eq:Omega1}.
\end{proof}

\begin{proof}[Proof of \eref{eq:Omega2}]
The sum of all  test projectors $P_{\bfr}$ with $\sum_k r_k=0\mod d$ can be expressed as
\begin{align}\label{eq:sumPr}
\sum_{\bfr}P_{\bfr}
&= d^{n-2}\openone + \frac{1}{d} \sum_{l=1}^{d-1}  \sum_{\bfr}  \prod^{n}_{k=1} \big(X_k Z_k^{r_k}\big)^l\nonumber\\
&= d^{n-2}\openone + \frac{1}{d^2} \sum^{d-1}_{l=1}  \sum_{s\in \bbZ_d}   \bigg[\sum_{r\in \bbZ_d} \omega^{-sr} (XZ^r)^l \bigg]^{\otimes n}\nonumber\\
&= d^{n-2}\openone + \frac{1}{d^2} \sum^{d-1}_{l=1}  \sum^{d-1}_{j=0}  \big( d |j+ l\>\<j| \big)^{\otimes n}\nonumber\\
&= d^{n-2} \biggl[ \openone+\sum_{j'\ne j}(|j'\>\<j|)^{\otimes n} \biggr],
\end{align}
which implies \eref{eq:Omega2}. The first equality is meaningful when $d$ is odd, in which case $(X_k Z_k^{r_k})^d=I$,
where $I$ is the identity operator on the Hilbert space of one qudit. The third equality follows from the following fact:
For each $s\in \bbZ_d$ and $l\in \{1,2,\dots,d-1\}$, we have
\begin{align}
&\sum_{r\in \bbZ_d} \omega^{-sr} (XZ^r)^l
= X^l \sum_{r\in \bbZ_d}  \omega^{r[l(l-1)/2-s]} Z^{rl}\nonumber\\
&= \sum_{j=0}^{d-1} |j+ l\>\<j| \bigg(\sum_{r\in \bbZ_d} \omega^{r[l(l-1)/2+jl-s]} \bigg).
\end{align}
The last term in the parentheses vanishes except when $l(l-1)/2+jl-s=0\mod d$, in which case it equals $d$.
If $d$ is an odd prime and  $l\neq0$, then the equation $l(l-1)/2+jl-s=0\mod d$ for each $s$ has a unique solution for $j\in \bbZ_d$, and the map from $s$ to the solution $j$ is one to one,
so the third equality in \eref{eq:sumPr} holds.

To clarify why the above proof does not work when $d$ is an odd number that is not a prime, suppose $l$ is a divisor of $d$.  Then the equation $l(l-1)/2+jl-s=0\mod d$
has multiple solutions  when $s$
is a multiple of $l$, while it has no solution otherwise, so  the third equality in \eref{eq:sumPr} does not hold in this case. Therefore,  we need to assume that $d$ is an odd prime in order to construct an optimal protocol based on Pauli measurements.
\end{proof}

\section{\label{app:ProofOmega3}Proofs of \eqsref{eq:XW^h}{eq:Omega3}}
\begin{proof}[Proof of \eref{eq:XW^h}]
According to \eref{eq:2designGood}, we have
\begin{equation}
|\psi_{ht}\>=\frac{1}{\sqrt{d}}\sum_{j=0}^{d-1} \rme^{\rmi \theta_{htj}}|j\>=\frac{1}{\sqrt{d}}\sum_{j=0}^{d-1} \omega^{tj}\mu^{h\binom{j}{2}}|j\>.
\end{equation}
Therefore,
\begin{align}
&XW^h \sqrt{d}|\psi_{ht}\>=\sum_{j=0}^{d-2} \omega^{tj}\mu^{h\left[\binom{j}{2}+j\right]}|j+1\>\nonumber\\
&\quad +\omega^{-t}\mu^{h\left[\binom{d-1}{2}-(d-1)(d-2)/2\right]}|0\>\nonumber\\
&=\omega^{-t}\sum_{j=0}^{d-1} \omega^{tj}\mu^{h\binom{j}{2}}|j\>=\omega^{-t}\sqrt{d}|\psi_{ht}\>.
\end{align}	
It follows that $|\psi_{ht}\>$ is an eigenstate of $XW^h$ with eigenvalue $\omega^{-t}$, which implies \eref{eq:XW^h}.
	
Alternatively, \eref{eq:XW^h} can be proved as follows.
\begin{align}\label{eq:XW^hproof}
&\sum_{t\in \bbZ_d} \omega^{-t}(|\psi_{ht}\>\<\psi_{ht}|)\nonumber\\
&=\frac{1}{d}\sum_{j,j'=0}^{d-1}(|j'\>\<j|) \bigg(\rme^{\rmi \pi h(j'-j)(j'+j-1)/m}\sum_{t\in \bbZ_d}\omega^{ t(j'-j-1)}\bigg)\nonumber\\
&=\sum_{j=0}^{d-1}(|\hat{j}\>\<j|)  \big(\rme^{\rmi \pi h(\hat{j}-j)(\hat{j}+j-1)/m}\big)\nonumber\\
&=\mu^{-h(d-1)(d-2)/2}(|0\>\<d-1|)+\sum_{j=0}^{d-2}\mu^{hj}(|j+1\>\<j|)\nonumber\\
&=XW^h,
\end{align}
where $\hat{j}=j+1$ if $j\leq d-2$ and $\hat{j}=0$ if $j=d-1$.
\end{proof}

\begin{proof}[Proof of \eref{eq:Omega3}]
	The sum of all test projectors $P_{\bfh}$ that satisfy the condition $\sum_k h_k=0\mod m$ can be expressed as
	\begin{align}
	\sum_{\bfh}P_{\bfh}
	&= \frac{m^{n-1}}{d}\openone + \frac{1}{d} \sum_{l=1}^{d-1}  \sum_{\bfh}  \prod^{n}_{k=1} \big(X_k W_k^{h_k}\big)^l\nonumber\\
	&= \frac{m^{n-1}}{d}\openone + \frac{1}{dm} \sum^{d-1}_{l=1} \sum_{s\in \bbZ_m} \bigg[\sum_{h=1}^{m} \mu^{-sh} (XW^h)^l \bigg]^{\otimes n} \nonumber\\
	&= \frac{m^{n-1}}{d}\openone + \frac{1}{dm} \sum^{d-1}_{l=1} \sum_{j=0}^{d-1} \big( m |j+ l\>\<j| \big)^{\otimes n}\nonumber\\
	&= \frac{m^{n-1}}{d} \biggl[ \openone+\sum_{j'\ne j}(|j'\>\<j|)^{\otimes n} \biggr], \label{eq:sumPh}
	\end{align}
	which implies \eref{eq:Omega3}. To derive the  third equality, for each  $s\in\bbZ_m$ and $l=1,2,\dots,d-1$, define
\begin{align}
f(s,l):&=\sum_{h=1}^{m} \mu^{-sh} (XW^h)^l.
\end{align}	
Thanks  to \eref{eq:XW^h} in the main text or \eref{eq:XW^hproof}, we have
	\begin{align}\label{eq:f(s,l)}
f(s,l)&=\sum_{h=1}^{m} \mu^{-sh} \sum_{t\in \bbZ_d} \omega^{-tl}|\psi_{ht}\>\<\psi_{ht}| \nonumber\\
	&=\frac{1}{d}\sum_{j,j'=0}^{d-1}(|j'\>\<j|) \bigg(\sum_{t\in \bbZ_d} \omega^{t(j'-j-l)}\bigg)   \nonumber\\
	&\quad\,        \times\bigg(\sum_{h=1}^{m}\mu^{h[(j'-j)(j'+j-1)/2-s]}\bigg)  \nonumber\\
	&=\sum_{j=0}^{d-1}(|\hat{j}\>\<j|) \bigg(\sum_{h=1}^{m}\mu^{h[g(j,l,d)-s]}\bigg),
	\end{align}
where
\begin{align}
\hat{j}:=&\begin{cases}
j+l          &j+l\leq d-1,\\
j+l-d \qquad &j+l\geq d,
\end{cases}\\
g(j,l,d):=&\frac{1}{2}(\hat{j}-j)(\hat{j}+j-1). \label{eq:gjld}
	\end{align}
	The last term in the parentheses in \eref{eq:f(s,l)} vanishes except when
\begin{equation}\label{eq:gjlds}
g(j,l,d)-s=0\mod m,
\end{equation}	
in which case it is equal to $m$.
For given $l$ and $j$, note that \eref{eq:gjlds} has a unique solution for $s\in \bbZ_m$. Conversely, for each  $l\in\{1,2,\dots,d-1\}$ and  $s\in \bbZ_m$,  \eref{eq:gjlds} has at most one solution for $j\in\{0,1,\dots,d-1\}$ by \lref{lem:gjlInjective} below given that  $m\geq\lceil\frac{3}{4}(d-1)^2\rceil$. This result implies the third equality in \eref{eq:sumPh} and completes the proof of \eref{eq:Omega3}.
\end{proof}

\begin{lemma}\label{lem:gjlInjective}
Let
\begin{equation}
g_m(j,l,d):=g(j,l,d)\mod m,
\end{equation}
where $g(j,l,d)$ is defined in \eref{eq:gjld}. Suppose $d\geq 3$, $m\geq\lceil\frac{3}{4}(d-1)^2\rceil$, and  $l\in\{1,2,\dots,d-1\}$; then  $g_m(j,l,d)$ is injective in $j$ for $j\in \{0,1,\ldots, d-1\}$.
\end{lemma}
This lemma follows from Proposition~4.3 in \rcite{RoyS07}. Here we present a self-contained proof for completeness.
\begin{proof}
When $j\in\{0,\dots,d-l-1\}$, the function $g(j,l,d)$ is monotonically increasing in $j$, and we have
\begin{equation}
0\leq l(l-1)/2 \leq g(j,l,d)\leq l(2d-l-3)/2 <m
\end{equation}
given that  $m\geq\lceil\frac{3}{4}(d-1)^2\rceil$.
When $j\in\{d-l,\dots,d-1\}$ by contrast, $g(j,l,d)$ is monotonically decreasing in $j$, and we have
\begin{align}
&-m <  (l-d)(d+l-3)/2 \leq g(j,l,d)\nonumber \\
&\leq(l-d)(d-l-1)/2\leq0.
\end{align}
In addition, it is straightforward to verify that
\begin{equation}
l(2d-l-3)/2<(l-d)(d+l-3)/2+m.
\end{equation}
Therefore, the two sets of numbers $\{g_m(j,l,d)\}_{j=0}^{d-l-1}$ and $\{g_m(j,l,d)\}_{j=d-l}^{d-1}$ have no intersection; moreover,  all the numbers  $g_m(0,l,d), g_m(1,l,d), \dots, g_m(d-1,l,d)$ are distinct, which confirms \lref{lem:gjlInjective}.
\end{proof}

\section{\label{app:StratGHZ-like}Alternative optimal protocol for verifying GHZ-like states}
In the main text we proposed an optimal strategy for verifying the GHZ-like state $|\xi\>=\sum_{j=0}^{d-1}\lambda_j|j\>^{\otimes n}$ based on one-way LOCC, which requires only $1+2^{n-1}$ distinct tests when $d=2$ and
$1+m^{n-1}$ distinct tests with $m\geq\lceil\frac{3}{4}(d-1)^2\rceil$  when $d\geq 3$. Here we propose an  alternative optimal protocol  using much fewer measurement settings, assuming that  the local dimension $d$ is an odd prime. In addition, for each test, all parties except for one of them can perform Pauli measurements as in the case of qubits.
The underlying idea is similar to the construction of the strategy $\Omega_{\5}$ in Sec.~\ref{sec:GHZlikeOpt}.

For each string $\bfr\in \bbZ_d^n$ with $\sum_k r_k=0\mod d$, define the test projector
\begin{equation}
P'_{\bfr}:=\big(I^{\otimes (n-1)}\otimes M\big)P_{\bfr}\big(I^{\otimes (n-1)}\otimes M\big),
\end{equation}
where $M:=\sqrt{d}\diag(\lambda_0,\lambda_1,\dots,\lambda_{d-1})$, and $P_{\bfr}$ is the test projector given in \eref{eq:Pother(d=p)}. Then $P'_{\bfr}$ can be realized by adaptive local projective measurements as described in  Sec.~\ref{sec:GHZlikeOpt}.
According to \eref{eq:sumPr}, we have
\begin{align}
\frac{1}{d^{n-1}}\sum_{\bfr}P_{\bfr}&=\frac{1}{d} \biggl[\openone+\sum_{j\ne j'}(|j'\>\<j|)^{\otimes n}\biggr].
\end{align}
As a corollary,
\begin{align}
&\frac{1}{{d^{n-1}}}\sum_{\bfr}P'_{\bfr}
= \big(I^{\otimes (n-1)}\otimes M\big)\bigg(\frac{\sum_{\bfr}P_{\bfr}}{{d^{n-1}}}\bigg)\big(I^{\otimes (n-1)}\otimes M\big)\nonumber\\
&=|\xi\>\<\xi|+I^{\otimes (n-1)}\otimes \rho_{n}-\sum_{j=0}^{d-1}\lambda_j^2(|j\>\<j|)^{\otimes n},\label{eq:PrpSum}
\end{align}
where  $\rho_{n}=\sum_{j=0}^{d-1}\lambda_j^2|j\>\<j|$ is the reduced state for party $n$. Note that
the right-hand side in \eref{eq:PrpSum} is identical to its counterpart in \eref{eq:Pi1}.

Suppose we perform the test $P_0$ with probability $p$ and the other tests $P'_{\bfr}$ with probability $(1-p)/d^{n-1}$ each.
Then the verification operator reads
\begin{equation}
\Omega'_{\5}=p P_0+\frac{1-p}{d^{n-1}}\sum_{\bfr}P'_{\bfr},
\end{equation}
and we have
\begin{equation}
\beta(\Omega'_{\5})=\beta(\Omega_{\5})
=\max\{p,(1-p)\lambda_0^2\}\geq\frac{\lambda_0^2}{1+\lambda_0^2}
\end{equation}
as in \eref{eq:betaOmega4} in the main text.
The lower bound is attained when $p={\lambda_0^2}/(1+\lambda_0^2)$, in which case we can achieve the maximum spectral gap $\nu(\Omega'_{\5})=1/(1+\lambda_0^2)$ as in \eref{eq:nuOmega5}.
When $d$ is an odd prime,  therefore, $1+d^{n-1}$ distinct tests are sufficient for constructing a strategy that is equivalent to $\Omega_{\5}$ in Sec.~\ref{sec:GHZlikeOpt}, which is optimal among all strategies based on one-way LOCC. Nevertheless, the efficiency can be improved further  by virtue of more communications as employed in the construction of  $\Omega_{\6}$ in Sec.~\ref{sec:GHZlikeTwoWay}.

\end{document}